\documentclass[USenglish]{lipics-v2019}
\pdfoutput=1 

\nolinenumbers
\newtheorem{conj}{Conjecture}

\newtheorem{observation}{Observation}

\usepackage{thm-restate}
\usepackage{hyperref}
\usepackage{wrapfig}
\usepackage{epsfig}
\usepackage{xspace}
\usepackage{bbm}
\usepackage{algorithm}
\usepackage{algorithmic}

\newcommand{\etal}{\emph{et al.\/}}

\newcommand{\R}{\mathbb{R}}

\newcommand{\D}{\mathcal{D}}
\renewcommand{\L}{\mathcal{L}}
\renewcommand{\P}{\mathcal{P}}
\newcommand{\U}{\mathcal{U}}

\newcommand{\Tg}{T_{\mbox{\scriptsize global}}}
\newcommand{\vol}{{\mbox{Vol}}}

\newcommand{\mdef}[1]{{\it #1}}

\newcommand{\eps}{\varepsilon}

\renewcommand{\c}[1]{\ensuremath{\EuScript{#1}}}

\newcommand {\mm}[1] {\ifmmode{#1}\else{\mbox{\(#1\)}}\fi}

\title{Independent Range Sampling, Revisited Again}

\author{Peyman Afshani}{Aarhus University}{peyman@cs.au.dk}{}{supported by DFF (Det Frie Forskningsr\" ad) of Danish Council for Indepndent Reserach under grant ID DFF$-$7014$-$00404.}
\author{Jeff M. Phillips}{University of Utah}{jeffp@cs.utah.edu}{}{supported by NSF CCF-1350888, CNS-1514520, CNS-1564287, IIS-1816149, and in particular ACI-1443046.  Part of the work was completed while visiting the Simons Institute for Theory of Computing.}
 
\authorrunning{Afshani \& Phillips}

\Copyright{Peyman Afshani and Jeff M. Phillips}%TODO mandatory, please use full first names. LIPIcs license is "CC-BY";  http://creativecommons.org/licenses/by/3.0/
%TODO mandatory; please add comma-separated list of keywords

\begin{document}
%\begin{titlepage}
\maketitle

%\begin{center} \today \end{center}
\begin{abstract}
We revisit the range sampling problem:
  the input is a set of points where each point is associated with a real-valued weight.
  The goal is to store them in a structure such that given a query range and an integer
  $k$, we can extract  $k$ independent random samples from the points inside the query range, where
  the probability of sampling a point is proportional to its weight. 

  This line of work was initiated in 2014 by Hu, Qiao, and Tao and it was later followed up by Afshani
  and Wei. 
  The first line of work mostly studied unweighted but dynamic version of the problem in  one dimension
  whereas the second result considered the static weighted problem in one dimension as well as the unweighted
  problem in 3D for halfspace queries. 

  We offer three main results and some interesting insights that were missed by the previous work: 
  We show that it is possible to build efficient data structures
  for range sampling queries if we allow the query time to hold in expectation (the first result),
  or obtain efficient worst-case query bounds by allowing the sampling probability to be approximately
  proportional to the weight (the second result).
  The third result is a conditional lower bound that shows essentially one of
  the previous two concessions is needed.
  For instance, for the 3D range sampling queries, the first two results give efficient data structures
  with near-linear 
  space and polylogarithmic query time whereas the lower bound shows with near-linear space the worst-case
  query time must be close to $n^{2/3}$, ignoring polylogarithmic factors. 
  Up to our knowledge, this is the first such major gap between the expected and worst-case query time of
  a range searching problem.
\end{abstract}

%\end{titlepage}

%a%%%%%%%%%%%%%%%%%%%%%%%%%%%%%%%%%%%%%%%%%%%%%%%%%%%
\section{Introduction}
In range searching, the goal is store a set $P$ of points in a data structure such that given a query range, we can answer certain 
questions about the subset of points inside the query range. 
The difficulty of the range searching problem, thus depends primarily on the shape of the query as well as types of questions
that the data structure is able to answer. 
These range searching questions have been studied extensively and we refer the reader to the
survey by Agarwal and Erickson~\cite{ae98} for a deeper review of range searching problems. 

Let us for a moment fix a particular query shape.
For example, assume we are given a set $P \subset  \R^2$ to preprocess and at the query time we will be given a query
halfplane $h$.  % by the user.  
The simplest type of question is an {\em emptiness} query where we simply want to report to the user
whether $h \cap P$ is empty or not.
Within the classical literature of range searching, the most general (and thus the most difficult) variant of range searching is
semigroup range searching where each point in $P$ is assigned a weight from a semigroup and at the query time the goal 
is to return the sum of the weights of the points of $h \cap P$.
The restriction of the weights to be from a semigroup is to disallow subtraction.
As a result, semigroup range searching data structures can answer a diverse set of questions.  %, and they are very flexible. 
Other classical variants of range searching lie between the emptiness and the semigroup variant.
In our example, the emptiness queries can be trivially solved with $O(n)$ space and $O(\log n)$ query time
whereas the semigroup variant can only be solved with $O(\sqrt{n})$ query time using $O(n)$ space. 
Finally, the third important variant, range reporting, where the goal is to output all the points in $P \cap h$, is often close
to the emptiness variant in terms of difficulty.
E.g., halfplane queries can be answered in $O(\log n + k)$ time where $k$ is the size of the output, using $O(n)$ space. 

\subparagraph{Sampling queries.}
Let $P$ be a large set of points that we would like to preprocess for range searching queries. 
Consider a query range $h$. 
Classical range searching solutions can answer simple questions such as the list of points inside $h$, or the number of them.
However, we immediately hit a barrier if we are interested in more complex questions, e.g.,
what if we want to know how a ``typical'' point in $h$ looks like? 
Or if we are curious about the distribution of the data in $h$. 
In general, doing more complex data analysis requires that we extract the list of all the points inside $h$ but this could
be an expensive operation. 
For questions of this type, as well as many other similar questions, it is very useful to be able to extract a (relatively) small
random sample from the subset of points inside $h$. 
In fact, range sampling queries were considered more than two decades ago within the context of database 
systems~\cite{hellerstein1997online,agarwal2013blinkdb,chaudhuri1998random,olkensurvey,olkenthesis}.  
Indeed the entire field of sample complexity, which provides the basis for statistics and machine learning, argues how any statistical quantity can be understood from an iid sample of the data.  Range sampling allows this literature to quickly be specified to data in a query range.  
However, many of these classical solutions fail to hold in the worst-case as they use R-trees or Quadtrees
whose performance depends on the distribution of the coordinates of the input points (which could be pretty bad).
Some don't even guarantee that the samples extracted in the future will be independent of the samples extracted in the past. 
For example, sometimes asking the same query twice would return the same samples.

\subparagraph{The previous results.}
Given a set of $n$ weights, one can sample one weight with proportional probability using
the well-known ``alias method'' by A. J. Walker~\cite{walker1974new}.
This method uses linear space and can sample a weight in worst-case constant time.
%but does not allow
%updating the weights.
%Using relatively simple tree-based methods, it is possible to update the weights in $O(\log n)$
%time and also extract samples in $O(\log n)$ time as well~\cite{dynamicAlias}.

The rigorous study of {\em independent range sampling} was initiated by Hu \etal~\cite{hu2014independent}.
They emphasized the requirement that any random sample extracted from the data structure should be
independent of all the other samples. 
They studied the one-dimensional version and for an unweighted point set and presented
a linear-sized data structure that could extract $k$ random samples in $O(\log n + k)$ time
and it could be updated in $O(\log n)$ time as well. 
A few years later, Afshani and Wei~\cite{AW2017} revisited the problem and
solved the one-dimensional version of the problem
for weighted points: they presented a linear-size data structure that could
answer queries in $O(\mbox{Pred}(n)+k)$ time
where $\mbox{Pred}(n)$ referred to the running time of a predecessor query
(often $O(\log n)$ but sometimes it could
be faster, e.g., if the input is indexed by an array then the predecessor query
can be answered trivially in $O(1)$ time).
They also studied the 3D halfspace queries but for unweighted points.  
Their main result was an optimal data structure of linear-size that could answer
queries in $O(\log n + k)$ time. 

\subparagraph{Our results.}
We provide general results for the independent range sampling problem on weighted input (wIRS), and design specific results for 3D halfspace queries.  
We show a strong link between the wIRS and the range max problem.  Namely, we show that the range max problem is at least as hard as wIRS, and also we provide a general formulation to solve the wIRS problem using range max.  
For halfspace queries in 3D, our framework gives a structure that uses $O(n \log n)$ space and has $O(\log^2 n +k)$ query time.
We improve the space complexity to $O(n)$ when the ratio of the weights is $n^{O(1)}$.  
This solution uses rejection sampling, so it only provides an expected query time bound.  
To compensate, we provide another solution that has worst-case query time, but
allows the points to be sampled within a $(1 \pm \eps)$ factor of their desired
probability, and may ignore points that would be sampled with probability less
than $\gamma/n$ for $\gamma < \eps < 1$.  This structure requires $O(n \log n)$
space, and has a worst-case query time of $O(\log(n/\gamma) (\log n + 1/\eps^3)
+ k)$.  
%When $\eps, \gamma = \Omega(1)$, then this matches the expected query time of the previous structure.  

Finally, we show a conditional lower bound when we enforce worst-case query time and exact sampling probabilities, in what we call the separated algebraic decision tree (SAD) model.  This model allows any decision tree structure that compares random bits to algebraic functions of a set of input weights.  
In this model, we show wIRS is as hard as the range sum problem, which is conjectured to be hard.  In particular, if the best known solution to the range sum problem for halfspaces in 3D is optimal, then the wIRS problem would require $\Omega(n^{2/3- o(1)})$ query time if it uses near-linear space.  
This provides the first such separation between expected $O(\log^2 n + k)$ and worst-case $\Omega(n^{2/3 - o(1)})$ query time for a range searching problem that we are aware of.

% !TeXroot=range-sampling.tex

\section{A Randomized Data Structure}\label{sec:rand}
In this section, we show that if we allow for the query bound to hold in expectation, then the range
sampling problem can be solved under some general circumstances. 
Intuitively, we show that we need two ingredients: one, a data structure for range maximum queries,
and two, a data structure that can sample from a weighted set of points under the assumptions that the weights 
are within a constant factor of each other. 
Furthermore, with an easy observation, we can also show that the range sampling problem is at least as hard as
range maximum problem. 
We consider the input as a general set system $(X, \c{R})$.
We assume the input is a set $X$ of $n$ data elements (e.g., points) and we consider queries to be elements of 
a set of ranges $\c{R}$ where each $R \in \c{R}$ is a subset of $X$. 
The set $\c{R}$ is often given implicitly, and for example, if $X$ is a set of points in 
$\R^3$, $\c{R}$ could be the set of all $h \cap X$ where
$h$ is a halfspace in 3D.
Note that our model of computation is the real RAM model.

\begin{definition}[The Range Maximum Problem]
  Let $X$ be a data set, s.t., each element $x \in X$ is assigned a real-value weight $w(x)$.
  The goal is to store $X$ in a structure, s.t., given  a range $R \in \c{R}$,
  we can find the element in $R$ with the maximum weight. 
\end{definition}

Given a weighted set $X$, for any subset $Y \subset X$, we denote by $w(Y)$ the sum
$\sum_{x\in Y} w(x)$.
First we observe that range sampling is at least as hard as the range maximum problem. 
\begin{restatable}{lemma}{lemreduction}\label{lem:reduction}
  Assume we can solve range sampling queries on input $(X,\c{R})$ and for any weight assignment
  $w \colon X \to \R$ using $S(|X|)$ space and $Q(|X|)$ query time where the query only returns one
  sample. 
  Then, given the set $X$ and a weight function $w'\colon X\rightarrow \R$, 
  we can store $X$ in a data structure of size $S(|X|)$ such that
  given a query $R$, we can find the data element in $R$ with the maximum weight in
  $Q(|X|)$ time, with high probability. 
\end{restatable}
\begin{proof}
    Let $x_1, \cdots, x_n$ be the list of input elements in $X$ sorted by their weight function $w'$.
    We then assign the element $x_i$  a new weight $w(x_i) =  n^{ci}$, for a large enough constant $c$, 
    and store them in the data structure for the range sampling problem. 
    This takes $S(|X|)$ space. 
    Given a query $R$ for the range maximum problem, we sample one element $x_i$ from the range $R$.
    We have $w(x_i) = n^{ci}$ and the total weight of all the other points in $R$ can be at most
    $n\cdot n^{c(i-1)}$, i.e., we find the element with the maximum weight with high probability.
\end{proof}

Next, we show that weighted range sampling can be obtained from a combination of a range maximum data structure
and a particular form of weighted range sampling data structure for almost uniform weights.

\begin{lemma}\label{lem:general}
    Let $(X,\c{R})$ be an input to the range sampling problem.
    Assume, we have a structure for the range maximum queries that uses $O(S_m(|X|))$ space and with query time
    of $O(Q_m(|X|))$.
    Furthermore, assume for any subset $X' \subset X$ 
    we can build a structure $\D_s(X')$ that uses $O(S_s(|X'|))$ space and given a query $R \in \c{R}$, it does the following:
    it can return $w(Y)$ for a subset $Y \subset X'$ with the property that $R \cap X' \subset Y$, and
    $|Y| = O(|R \cap X'|)$ and furthermore, the structure can extract $k$
    random samples from $Y$ in $O(Q_s(|X|) + k)$ time.

  Then, we can answer range sampling queries using $O(S_m(|X|) + S_s(|X|))$ space and with expected query time
  of $O(Q_m(|X|) + Q_s(|X|)\log |X|+ k)$.
\end{lemma}
\begin{proof}
Let $n = |X|$.
  We store $X$ in a data structure for the range maximum queries. 
  We partition $X$ into subsets $X_i \subset X$ in the following way.
  We place the element $x_1$ with the largest weight in $X_1$ and then we add to $X_1$ any element whose
  weight is at least $w(x_1)/2$ and then recurse on the remaining elements.
  Observe that for all $x,x' \in X_i$ we have $w(x)/w(x') \in (1/2, 2]$.
  Thus, the weight function $w$ is almost uniform on each $X_i$. 
  We store $X_i$ in a data structure $\D_s(X_i)$.

  Next, we build a subset sum information over the total weight
  $W(X_i)$ of the (disjoint) union of all subsets $X_{i'}$ with $i' \geq i$,
  that is, $W(X_i) = \sum_{j\ge i}w(X_j)$.
  Consider a query $R \in \c{R}$.

\subparagraph{Step 1: Use the Range Maximum Structure.}
  Issue a single range-max query on $X$ for the query range $R \in \c{R}$.
  Let $x$ be the answer to the range-max query, assume $x \in X_i$.

\subparagraph{Step 2: top-level alias structure.}
Having found $i$, we identify the smallest index $i'\ge i$ such that the 
maximum weight $w(x')$ for $x' \in X_{i'}$ and the minimum weight $w(x)$ for $x\in X_i$ satisfy $w(x)/w(x') > n^2$.  
As the weights in the sets $X_i$ decreases geometrically, we have $i' - i = O(\log n)$.
Then, for each $X_j$ with $j \in [i,i')$, we use the data structure $\D_s(X_j)$ to identify the set $Y_j$ such that
$Y_j$ contains the set $R \cap X_j$. 
This returns the values $w(Y_j)$, and the total running time of this step is $O(Q_s(n)\log n)$.

We build a top-level alias structure on $i'-i + 1$ values: all the $i'-i$ values $w(Y_j)$, $i \le j  < i'$, as well as the value 
$W(X_{i'})$. 
Let $T = W(X_{i'}) + \sum_{j=i}^{i'-1}w(Y_j)$.
This can be done in $O(\log n)$ time as $i'-i  = O(\log n)$.

\subparagraph{Step 3: Extracting samples.}
To generate $k$ random samples from $R$, we first sample a value using the top-level alias structure. 
This can result in two different cases:
\subparagraph{Case 1.} It returns the value $W(X_{i'})$. 
Let $X_{i'+} =\cup_{j\ge i'}X_j$. 
In this case, we sample an element from the set $R \cap X_{i'+}$, by building
an alias structure on $X_{i+}$ in $O(n)$ time. 
The probability of sampling an element $x_j \in X_{i'+}$ is  set exactly to $\frac{w(x_j)}{W(X_{i'})}$.
Note that these probabilities of $x_j \in X_{i'+} \cap R$ do not add up to one which means 
the sampling might fail and we might not return any element.
If this happens, we go back to the top-level alias structure and try again.
Notice that $R$ contains at least one element $x_i$ from $X_i$, 
and that for any $x \in X_j$, $j\ge i'$ have $w(x) \le w(x_i)/n^2$ which implies
$W(X_{i'}) \le w(x_i)/n$. 
Thus, this case can happen with probability at most $1/n$, meaning, even if we spend
$O(n)$ time to answer the query, the expected query time is $O(1)$.
\subparagraph{Case 2.} It returns a value $w(Y_j)$, for $i \le j < i'$.
We place $Y_j$ into a list for now. 
At some later point (to be described) we will extract a sample $z$ from $Y_i$.
    If $z$ happens to be inside $R$, then we return $z$, otherwise, the sampling fails and we go back to the
    top-level structure.

We iterate until $k$ queries have been pooled. 
Then, we issue them in $O(\log n)$ batches to data structure $\D_s(X_j)$, $i \le j < i'$.
Ignoring the failure events in case (2), processing a batch of $k$ queries will take
$O(Q_s(n)\log n + k)$ time. 
Notice that each iteration of the above procedure will 
succeed with a constant probability: case (1) is very unlikely (happens with probability less than $1/n$)
and for case (2)  observe that we have $w(Y_i) = O(w(R\cap X_i))$ and thus
each query will succeed with constant probability.
As a result, in expectation we only issue a constant number of batches of size $k$ to extract $k$ random samples. 

It remains to show that we sample each element with the correct probability. 
The probability of reaching case (1) is equal to $\frac{W(X_{i'})}{T}$. 
Thus, the probability of sampling an element $x_j \in X_{i'+}$ is  equal to
$\frac{W(X_{i'})}{T}\cdot \frac{w(x_j)}{W(X_{i'})} = \frac{w(x_j)}{T}$. 
The same holds in case (2) and the probability of sampling an element $x_j \in X_j$, $i \le j <i'$ is
$\frac{w(x_j)}{T}$.
Thus, conditioned on the event that the sampling succeeds, each element is sampled
with the correct probability.
\end{proof}

\section{3D Halfspace Sampling}\label{sec:h}
\subsection{Preliminaries}
In this section, we consider random sampling queries for 3D halfspaces. 
But we first need to review some preliminaries. 
\begin{restatable}{lemma}{lemmaxu}\label{lem:maxu}
  Let $P$ be a set of $n$ points in $\R^3$.
  Let $P= P_1 \cup \cdots \cup P_t$ be a partition of $P$ into $t$ subsets.
  We can store $P$ in a data structure of size $O(n \log t)$ such that given 
  a query halfspace $h$, we can find the smallest index $i$ such that
  $P_i \cap h \not = \emptyset$ in $O(\log n \log t)$ time. 
\end{restatable}
\begin{proof}
    We consider the dual problem where $H_i$ is the set of hyperplanes dual to 
    $P_i$ and the goal is to store them
    in a data structure such that given a query point $q$, we can find the smallest index $i$ such that there exists
    a halfspace of $H_i$ that passes below $q$. 

    Let $H_\ell = H_1 \cup \cdots \cup H_{t/2}$ and $H_r = H_{t/2+1} \cup \cdots H_t$.
    We compute the lower envelope of $H_\ell$, and store its projection in a point location data structure and then recurse
    on $H_\ell$ and $H_r$.
    The depth of the recursion is $O(\log t)$ and each level of the recursion consumes $O(n)$ space and thus the total storage
    is $O(n \log t)$.

    Given a query $q$, we have two cases: if a halfspace of $H_\ell$ passes below $q$, then the answer is obviously in $H_\ell$ so
    we recurse there;  otherwise, no halfspace of $H_\ell$ passes below $q$ and thus, we can
    recurse on $H_r$. 
    This decision can be easily made using the point location data structure.
    The total query time is $O(\log t \log n)$.
\end{proof}

The following folklore result is a special case of the above lemma.
\begin{corollary}\label{cor:max}
  Let $P$ be a set of $n$ points in $\R^3$ where each point $p \in P$ is assigned a real-valued
  weight $w(p)$. 
  We can store $P$ in a data structure of size $O(n\log n)$ such that given a
  query halfspace $h$, we can find the point with maximum weight in $O(\log^2 n)$ time. 
\end{corollary}

We also need the following preliminaries.
Given a set $H$ of $n$ hyperplanes in $\R^3$, the \mdef{level} of a point $p$ is the number hyperplanes that pass below $p$.
The \mdef{$(\le k)$-level of $H$} (resp. \mdef{$k$-level of $H$}) 
is the closure of the subset of $\R^3$ containing points with level at most $k$ (resp. exactly $k$).
An \mdef{approximate $k$-level of $H$} is a surface composed of triangles (possibly infinite triangles) that lies
above $k$-level of $H$ but below $(ck)$-level of $H$ for a fixed constant $c$. 
For a point $q \in \R^3$, we define the conflict list of $q$ with respect to $H$ as the subset of  hyperplanes in 
$H$ that pass below $q$ and we denote this with $\Delta(H,q)$.
Similarly, for a triangle $\tau$ with vertices $v_1, v_2$, and $v_3$, we define $\Delta(H,\tau) = \Delta(H,v_1) \cup \Delta(H,v_2) \cup   \Delta(H,v_3)$.
One of the main tools that we will use is the existence of small approximate levels.
This follows from the existence of shallow cuttings together with some geometric observations.
\begin{restatable}{lemma}{lemshallow}\label{lem:shallow}
    For any set $H$ of $n$ hyperplanes in $\R^3$, and any parameter $1 \le k \le n/2$, there exists an approximate
    $k$-level which is a convex surface consisting of $O(n/k)$ triangles.

    Furthermore, we can construct a hierarchy of approximate $k_i$-levels $\L_i$, 
    for $k_i=2^i$ and $i=0, \cdots, \log n$,
    together with the list $\Delta(H,\tau)$ for every triangle $\tau \in \L_i$ in $O(n\log n)$ time and $O(n)$ space.
    Given a query point $q$, we can find an index $i$ in $O(\log n)$ time such that
    there exists a triangle $\tau \in \L_i$ that lies above $q$ such that
    $|\Delta(H,q)| = O(|\Delta(H,\tau)|)$.
\end{restatable}
\begin{proof}
  Shallow cuttings were first introduced by 
	Matou\v sek~\cite{Matousek.reporting.points} and later
	Chan~\cite{c00} observed that we can work with the triangulation of the convex hull of Matou\v sek's construction.
  As a result, shallow cuttings could be represented as convex surfaces formed by triangles. 
	Ramos~\cite{Ramos.SOCG99} offered a randomized $O(n\log n)$ time construction algorithm that could
  build a hierarchy of shallow cuttings in $O(n\log n)$ time possibly together with the conflict lists
  (i.e., $\Delta(H,\tau)$ for any triangle). 
  This was recently made deterministic~\cite{chanshallow}.

  Regarding the query part, it is known that we can store $H$ in 
  a data structure of size $O(n)$ such that given a query point $q$, we can
  find a constant factor approximation $\tilde k$ such that $|\Delta(H,q)|\le \tilde{k} = O(|\Delta(H,q)|)$ 
  in $O(\log n)$ worst-case time~\cite{AHZ.UB.CGTA}.
  We can then set $i=\log(\tilde{k})+O(1)$ and note that his method also finds $\tau$.
\end{proof}
We will also use the following results. 
     
\begin{theorem}[The Partition Theorem]\cite{Matousek.efficient.par}
	Given a set $P$ of $n$ points in 3D and an integer $0 < r \le n/2$,
	there exists a partition of $P$ into $r$ subsets $P_1, \cdots, P_r$, each of size
	$\Theta(n/r)$, where each $P_i$ is enclosed by a tetrahedron $T_i$, s.t.,
	any hyperplane crosses $O(r^{2/3})$ tetrahedra.

%	If $r\le n^\eps$ for a sufficiently small constant $\eps>0$, the
%	partition and tetrahedra can be constructed in $O(n\log r)$ time.
\end{theorem}

\begin{restatable}{lemma}{lemtree}\label{lem:tree}
  Let $T$ be tree of size $n$ where each leaf  $v$ stores a real-valued non-negative weight $w(v)$.
  We can build a data structure s.t., given an internal node $u \in T$ at the query time, we can independently sample
  a leaf with probability proportional to its weight in the subtree of $u$.
  The data structure uses $O(n)$ space and it can answer queries in $O(1)$ worst-case time. 
\end{restatable}
\begin{proof}
  Let $A$ be an array of the leafs obtained using the DFS ordering of the leafs.
  Observe that for any internal node $u$, the leafs in the subtree of $u$ correspond to a
  contiguous interval of $A$.
  Thus, the problem reduces to one-dimensional range sampling queries. 
  Afshani and Wei~\cite{AW2017} showed that these queries can be solved in $O(1)$ time plus the time
  it takes to answer predecessor queries. In our problem, since there are only $n$ different possible queries
  (one for each internal node), we can simply store a pointer from each internal node $u$ to the location of its
  predecessor in array $A$. 
\end{proof}

Afshani and Wei~\cite{AW2017} consider unweighted sampling for 3D halfspace queries.
We next use the following technical result of theirs.  
%lemma of theirs to do a slight strengthening of their results for inputs with weights. 
\begin{lemma}\label{lem:AW}
  Let $H$ be a set of $n$ hyperplanes in 3D. 
  Let  $f(n) = (\log n)^{c \log\log n}$ where $c$ is a large enough constant.
  We can build a tree $\Tg$ with $n$ leafs where each hyperplane is stored in one leaf such that the following holds:
  Given a point $q \in \R^3$ with level $k$ where $k \ge f(n)$, we can find
  $k'=O(k/\log^2 n)$ internal nodes $u_1, \cdots, u_{k'}$ in $\Tg$ such that 
  $\Delta(H,q) = \Tg(u_1) \cup \cdots \cup \Tg(u_{k'})$ where $\Tg(u_i)$ is the set of hyperplanes stored in the subtree of $u_i$.
\end{lemma}
Unfortunately, the above lemma is not stated explicitly by Afshani and Wei, however, $\Tg$ is the ``Global Structure'' that
is described in~\cite{AW2017} under the same notation.

\subsection{A Solution with Expected Query Time}
\label{sec:expected}
We now observe that we can use Lemma~\ref{lem:general} to give a data structure for
weighted halfspace range sampling queries in 3D. 
We first note that using Corollary~\ref{cor:max} and by building a hierarchy of shallow
cuttings, i.e., building approximate $k_i$-levels for $k_i = 2^i$, $i=0, \cdots, \log n$,
we can get a data structure with $O(n\log n)$ space that can answer queries in
$O(\log^2 n + k)$ query time.
Furthermore, as Lemma~\ref{lem:reduction} shows, our problem is at least as hard as the halfspace range maximum
problem which currently has no better solution than $O(n\log n)$ space and $O(\log^2 n)$ expected query time. 
Thus, it seems we cannot do better unless we can do better for range maximum queries, a problem that seems
very difficult.

However, the reduction given by Lemma~\ref{lem:reduction} is not completely satisfying 
since we need to create a set of weights that are exponentially distributed.
As a result, it does not capture a more ``natural'' setting where the ratio between the largest
and the smallest weight is bounded by a parameter $U$ that is polynomial in $n$.
Our improved solution is the following which shows when $U = n^{O(1)}$ we can in fact
reduce the space to linear. 
\begin{restatable}{theorem}{thmexp}\label{thm:3dexp}
  Let $P$ be a set of $n$ weighted points in $\R^3$, where the smallest weight is 1 and the largest weight is $U$.
  We can store $P$ in a data structure of size $O(n \min\{\log n , \log\log_n U\})$ such that
  given a query halfspace and an integer $k$, we can extract $k$ independent
  random samples from the points inside $h$ in $O(\log^2 n + k)$ time. 
\end{restatable}
\begin{proof}[Proof summary]
    Let $X$ be the set of hyperplanes dual to $P$.
    We partition $X$ into sets $X_1, \dots, X_t$ where the weights of the elements in 
    $X_i$ are larger than those of $X_{i+1}$ and for two hyperplanes $h, h' \in X_i$ we have
    $w(h)/w(h') \in [1/n^2, n^2]$.
    Clearly, we have $t \le \min\{n, \log_n U\}$.
    Given a query point $q$, we can find the smallest index $i$ such that a hyperplane of $H_i$ passes
    below $q$, in $O(\log t \log n) = O(\log^2 n)$ time using Lemma~\ref{lem:maxu}.
    Also, by building a prefix sum, $W_i = \sum_{j \ge i} w(X_j)$, we can essentially
    focus on sampling from $X_i$, similar to Lemma~\ref{lem:general}.
    With a slight abuse of the notation, let us rename $X_i$ as $X$.
    Let $m = |X|$.
    We build a hierarchy of approximate levels by Lemma~\ref{lem:shallow}.

    We use Lemma~\ref{lem:AW}.
    Consider a triangle $\tau \in \L_i$ with vertices $v_1, v_2$, and $v_3$ where the level of each vertex is
    at most $O(k_i)$. 
    By the lemma, if $k_i \ge f(m)$, then we can find a representation of
    $\Delta(X,v_1)$, $\Delta(X,v_2)$, and $\Delta(X,v_3)$, each using $O(k_i/\log^2m)$ internal nodes.
    Since $\Delta(X,\tau) = \Delta(X,v_1) \cup\Delta(X,v_2) \cup  \Delta(X,v_3)$, we can 
    also represent $\Delta(X,\tau)$ using $O(k_i/\log^2m)$ internal nodes of $\Tg$. 
    At $\tau$ we store an alias data structure on the weights of the subtrees of $\Tg$ that represent $\Delta(X,\tau)$.
    This requires $O(k_i/ \log^2 m)$ space.
    Observe that at the query time, we can proportionally sample one of these subtrees and then using 
    Lemma~\ref{lem:tree} we can sample an element from the subtree.
    Thus, we can extract one random sample from $\Delta(X,\tau)$ in $O(1)$ worst-case time 
    using only $O(\Delta(X,\tau)/\log^2 m)$ space.
    The total size of the lists $\Delta(X,\tau)$ for all $\tau \in \L_i$ and all $i$ is $O(m\log m)$, 
    meaning, this in total will consume $O(m)$ space. 

    By Lemma~\ref{lem:shallow}, we can identify an index $i$ s.t., $q$ lies
    below a triangle $\tau \in \L_i$ with $|\Delta(H,\tau)| = O(|\Delta(H,q)|)$.
    It is clear that we can now sample from $\Delta(H,\tau)$ using tree $\Tg$.
    This satisfies all the conditions in  Lemma~\ref{lem:general} as 
    we can sample from $\Delta(H,\tau)$ in $O(1)$ time, after an initial $O(\log n)$ search time. 

    If $k_i < f(m)$, we can repeat a classical idea that was also used by Afshani
    and Wei and that is build an approximate 
    $f(m)$-level $\L$ and for each triangle $\tau \in \L$, we can repeat the above solution just one more time.
    We will be able to handle the indices $i$ such that $k_i \ge f(c(fm)) = o(\log m) = o(\log n)$. 
    However, if $\Delta(H,q) = o(\log n)$, then we can simply find all the hyperplanes passing below $q$ in 
    $O(\log n)$ time and build an alias data structure on them. 
\end{proof}

%%%%%%%%%%%%%%%%%%%%%%%%%%%%%%%%%%%%%%%%%%%%%%%%%%%%
\subsection{Worst-Case Time with Approximate Weights}

All of the previous data structures sample items exactly proportional to their weight, but rely on rejection sampling.  Hence, their query time is expected, and not worst case.  With small probability these structures may repeatedly sample items which are not in $h$, and then need to reset and try again with no bound on the worst-case
time.
To achieve worst-case query time, we need some modifications.
We allow for items to be sampled ``almost'' proportional to their weights, i.e., we introduce
a notion of approximation.
As we shall see in the next chapter, without some kind of approximation, our task is very likely impossible.

\subparagraph*{Problem definition.}
We consider an input set $X$ of $n$ points where $x_i$ has weight $w_i$ and we would like to store  $X$ in a data structure.
At the query time, we are given a halfspace $h$ and a value $k$ and we would like to extract $k$ random samples
from $X \cap h$. 
Let $w(h) = \sum_{x_i \in X \cap h} w_i$, and set two parameters $0 < \gamma < \eps < 1$.  
Ideally, we would like to sample each $x_i$ with probability $w_i/w(h)$.    
Instead we sample $x_i$ with probability $\rho_i$, if $w_i/w(h) \geq \gamma/n$ then 
\begin{align}
  (1-\eps) \frac{w_i}{w(h)} \leq  \rho_i \leq (1+\eps) \frac{w_i}{w(h)},\label{eq:err}
\end{align}
and if $w_i/w(h) < \gamma/n$ then we must have $\rho_i \leq (1+\eps) \frac{w_i}{w(h)}$.
That is, we sample all items within a $(1\pm \eps)$ factor of their desired probability, except for 
items with very small weight, which could be ignored and not sampled.
The smaller items are such that the sum of their desired probabilities is at most $\gamma$.

\subparagraph*{Overview of modifications.}
We start in the same framework as Section~\ref{sec:expected} and Lemma \ref{lem:general}, i.e., we partition $X$ into subsets $X_\ell$ by weight.
Then, we need to make the following modifications:
(1) We can now ignore small enough weights; 
(2) We can no longer use rejection sampling to probe into each set $X_\ell$, rather we need to collect a bounded number (a function $f(\eps)$) of candidates from all $X_\ell$ which is guaranteed to contain some point in the query $h$.  
%We can then build a dynamic alias structure on these $t \cdot f(\eps)$ points without replacement until we find one in $h$.  
(3) We can also no longer use rejection sampling within each $X_\ell$ to get candidates, instead we build a stratified sample via the partition theorem.  

Change (1) is trivial to implement. 
Remember that given query $h$, we first identify indices $i$ and $i'$ such that $X_i$ contains the largest weight in $h$ and the weights in $X_{i'}$ are a factor $n^2$ smaller. 
We now require weights in $X_{i'}$ to be a factor $n/\gamma$ smaller, and let $i' = i+ t$ where $t=O(\log (n/\gamma))$. 
We can now ignore all the remaining sets: 
the sets $X_{j}$ with $j \geq i+t$ will have weights so small that even if there are $\Omega(n)$ points within, the sum of their weights will be at most $\gamma$ times the largest weight.  
Since $\gamma < \eps$, this implicitly increases all of the other weights in each $X_\ell$ (for $\ell \in [i,i+t)$) by at most a factor $(1+\eps)$.

%instead of ensuring each $X_i$ has all weights within a factor $2$, we ensure their weights are within a factor $(1+\eps/3)$.  We here after will treat all points in each $X_i$ as having the same weight -- this results in a $(1+\eps/3)$ factor error.  
%And on a query we need to consider sets $X_i$ which are no smaller than the range-max item's weight, times $\gamma/3n$ (before it was within $1/n$ of the range-max item's weight); 
%the smaller weight points are then below the $\gamma/3n$ threshold and can be completely ignored.  Since $\gamma < \eps$, this implicitly increases the other weights (and hence sampling probabilities) by at most a $(1 + \eps/3)$ factor.  Both distortion factors result in no more than $(1+\eps/3)^2 < (1+\eps)$ distortion.   
%This change also increases the number of sets $X_i$ with nearly the same weight which areinspected on a query from $O(\log n)$ to $O(\log_{(1+\eps)} (n/\gamma)) = O(\frac{1}{\eps} \log (n/\gamma))$.  

We next describe how we can build a data structure on each $X_\ell$ to
generate $f(1/\eps)$ candidate points.  
Once we have these $t \cdot f(1/\eps)$
points, we can build two alias structures on them (they will come with
weights proportional to the probability they should be selected), and select
points until we find one in $h$.  
As a first pass, to
generate $k$ samples, we can repeat this $k$ times, or bring $k f(1/\eps)$
samples from each $X_\ell$.  
We will return to this subproblem to improve
the dependence on $k$ and $\eps$ by short-circuiting a low-probability event and reloading these points dynamically.

\subsubsection{Generating Candidate Points}
Here we will focus on sampling our set of candidates.
We will do this for every set $X_\ell$, $i \le \ell < i'$.
However, to simplify the presentation, we will assume that the input
is a set $X$ of $n$ points such that the weights of the points in $X$ are within
factor 2 of each other. 
We will sample $f(1/\eps)$ candidate
points from $X$ (representing a subset $X_\ell$) s.t.,
the set of candidates intersects with the query halfspace $h$.  
Each candidate will be sampled with a probability that is almost uniform, i.e., 
it fits within our framework captured by Eq.~\ref{eq:err}.

Let $H$ be the set of hyperplanes dual to points in $X$.
We maintain a hierarchical shallow-cutting of approximate levels
(Lemma~\ref{lem:shallow}) on $H$.
By Lemma~\ref{lem:shallow}, we get the following
in the primal setting (on $X$), given a query halfspace $h$:
We can build $O(|X|)$ subsets of $X$ where the subsets
have in total $O(|X| \log |X|)$ points.
Given a query halfspace $h$, in $O(\log n)$ time, we can find a subset $X'$
so that $X \cap h \subset X'$ and $|X'| =O( |X \cap h|)$.  
We now augment this structure with the following information on each such subset $X'$, without increasing the space complexity.  
We maintain an $r$-partition $(Z_1, \Delta_1), (Z_2, \Delta_2), \ldots (Z_{r'}, \Delta_{r'})$ on $X'$.  (That is, so $r' = \Theta(r)$, each subset $Z_j \subset \Delta_j$ and has size bound $|X'|/r \leq |Z_j| \leq 2|X'|/r$, and the boundary of any halfspace $h$ intersects at most $O(r^{2/3})$ cells $\Delta_j$.)
For each $Z_j$ we maintain an alias structure so in $O(1)$ time we can generate a random $s_j$ from the points within.  It is given a weight $W_j = \sum_{x \in Z_j} w(x)$.  In $O(r)$ time we can generate a weighted set $S = \{s_1, s_2, \ldots, s_{r'}\}$; this will be the candidate set. 
%$X'$ we maintain a sample $S = \{s_1, s_2, \ldots, s_{r'}\}$ so that each $s_j$ is a random sample from $Z_j$.  

%\jeff{Peyman, can you help bound the space complexity of this structure.}

%We first describe how to draw a single point from $X' \cap h$.  
%We maintain a dynamic alias data structure~\cite{Melhorn?} on the set $S$, each $s_j$ is weighted by $|Z_j|$.  We draw samples without replacement until we select one that is in $H$, in the worst case this may take $O(r)$ draws, and $O(r)$ \jeff{check} time.  

The sum of all weights of points within $h$ is $w(h) = \sum_{x \in X \cap h} w(x)$, and so we would like to approximately sample each $x \in X' \cap h$ with probability $w(x)/w(h)$.  

\begin{lemma}
Let $r = \Omega(1/\eps^3)$ and consider a candidate set $S$, and sample one point proportional to their weights.   For a point $x \in X' \cap h$, the probability $\rho_x$ that it is selected satisfies
\[
(1-\eps) \frac{w(x)}{w(h)} \leq \rho_x \leq (1+\eps) \frac{w(x)}{w(h)}.
\]
\end{lemma}
\begin{proof}
Of the $r'$ cells in $S$, classify them in sets 
as \emph{inside} if $\Delta_j \in h$, 
as \emph{outside} if $\Delta_j \cap h = \emptyset$, and
as \emph{straddling} otherwise.  
We can ignore the outside sets.  There are $O(r)$ inside sets, and $O(r^{2/3})$ straddling sets.  

For point $x \in S_j$ from an inside set, it ideally should be selected with probability $\frac{w(x)}{W_j} \cdot \frac{W_j}{w(h)} = \frac{w(x)}{w(h)}$.  Indeed it is the representative of $S_j$ with probability $\frac{w(x)}{W_j}$ and is give weight proportional to $W_j$ in the alias structure.  
We now examine two cases, that all representative points in the straddling sets are in $h$, and that none are; the probability $x$ is selected will be between the probability of these two cases, and the desired probability it is selected will also be between these two cases.  
Let $W_{\text{\em in}} = \sum_{S_j \text{ is \emph{inside}}} W_j$ and $W_{\text{\em str}} = \sum_{S_j \text{ is \emph{straddling}}} W_j$.  The probability $x$ is selected if it is the representative of $S_j$ is then in the range 
$[\frac{W_j}{W_{\text{\em in}} + W_{\text{\em str}}} , \frac{W_j}{W_{\text{\em in}}}]$.  
The ratio of these probabilities is 
$\frac{W_j}{W_{\text{\em in}}} \cdot \frac{W_{\text{\em in}} + W_{\text{\em str}}}{W_j} = \frac{W_{\text{\em in}} + W_{\text{\em str}}}{W_{\text{\em in}}} = 1 + \frac{O(r^{2/3})}{\Theta(r)} = 1 + O(r^{1/3})$.  
Setting $r = \Omega(1/\eps^3)$ ensures that these probabilities are within a $(1+\eps)$-factor of each other, and on all points from an inside set, are chosen with approximately the correct probability.  

For a point $x$ in a straddling set $S_j$, it should be selected with probability $\frac{w(x)}{W_j}$ and is selected with probability between $L_x = \frac{w(x)}{W_j} \frac{W_j}{W_{\text{\em in}} + W_{\text{\em str}}}$ and $U_x = \frac{w(x)}{W_j} \frac{W_j}{W_{\text{\em in}} + W_j}$.  As with points from an inside set, these are within a $(1+\eps)$-factor of each other if $r = \Omega(1/\eps^3)$.  And indeed since $W_{\text{\em in}} \leq w(h) \leq W_{\text{\em in}} + W_{\text{\em str}}$ then $L_x \leq \frac{w(x)}{w(h)} \leq U_x$, and the desired probability of sampling straddling point $x$ is in that range.    
\end{proof}

\subsubsection{Constructing the $k$ Samples}
To select $k$ random samples, the simplest way is to run this procedure $k$ times sequentially, generating $O(k/\eps^3)$ candidate points;  this is on top of $O(\log n)$ to identify the proper subset $X'$ from the shallow cutting, applied to all $t = O(\log (n/\gamma))$ weight partitions $X_\ell$.  

We can do better by first generating $O(1/\eps^3)$ candidate points per $X_\ell$, enough for a single random point in $h$.  Now we place these candidates in two separate alias structures along with the candidate points from the other $t$ structures.  There are $O(t/\eps^3)$ candidate points placed in an \emph{inside} alias structure, and $O(t/\eps^2)$ points placed in a \emph{straddling} alias structure.  
Now to generate one point, we flip a coin proportional to the total weights in the two structures.  If we go to the \emph{inside} structure, we always draw a point in $h$, we are done.  If we go to the straddling structure, we may or may not select a point in $h$.  If we do, we are done; if not we flip another coin to decide to go to one of the two structures, and repeat.  

It is easy to see this samples points with the correct probability, but it does not yet have a worst case time.  We could use a dynamic aliasing structure~\cite{hmm} on the straddling set, so we sample those without replacement, and update the coin weight.  However, this adds a $O(\log(\frac{1}{\eps} \log (n/\gamma)))$ factor to each of $O(t/\eps^2)$ steps which might be needed.  A better solution is to only allow the coin to direct the sampler to the straddling sets at most once; if it goes there and fails, then it automatically redirects to the alias structure on the inside sets which must succeed.  
This distorts the sampling probabilities, but not by too much since in the rejection sampling scheme, the probability of going to the straddling set even once is $O(\eps)$.  

\begin{lemma}
For any candidate point $s$ let $\rho_s$ be the probability it should be sampled, and $\rho_s'$ the probability it is sampled with the one-shot deterministic scheme.  Then $\rho_s \leq \rho_s' \leq (1+\eps) \rho_s$ if $s$ is an inside point and $(1-\eps) \rho_s \leq \rho_s' \leq \rho_s$ if $s$ is from a straddling set.  
\end{lemma}
\begin{proof}
The probability that the coin directs to the inside set is 
$\pi_\textit{in} = \frac{W_\textit{in}}{W_\textit{in} + W_\textit{str}} = \frac{1}{1+W_\textit{str}/W_\textit{in}} = \frac{1}{1+\Omega(\eps)} = 1-O(\eps)$.  
Let $w_\textit{str}$ be the probability of selecting a point inside $h$, given that the coin has directed to the straddling set; we only need that $W_\textit{str} \in [0,1]$.  

For a candidate point in the inside set $s$, the probability it is selected in the rejection sampling scheme is $\rho_s = \frac{w(s)}{W_\textit{in}} (1-O(\eps))$, and in the deterministic scheme is $
\rho_s' = \frac{w(s)}{W_\textit{in}} (\pi_\textit{in} + (1-\pi_\textit{in}) (1-w_\textit{str}))$ which is in the range $[\rho_s, \frac{w(s)}{W_\textit{in}}]$, and these have a ratio $1+O(\eps)$.  

Similarly, for a candidate point in the straddling set $s$, the probability it is selected in the rejection sampling scheme is 
\[
\rho_s = \frac{w(s)}{W_\textit{str}} (1-\pi_\textit{in})(1 + W_\textit{str} (1-\pi_\textit{in})(1 + \ldots)) =\frac{w(s)}{W_\textit{str}} (1-\pi_\textit{in})(1+O(\eps)) 
\]
and in the deterministic scheme is
$
\rho'_s = \frac{w(s)}{W_\textit{str}} (1-\pi_\textit{in}).
$
Thus $\rho'_s$ is in the range $[\rho_s(1-O(\eps)), \rho_s]$.  Adjusting the constant coefficients in $\eps$ elsewhere in the algorithm completes the proof.  
\end{proof}

Now to generate the next independent point (which we need $k$ of), we do not need to re-query with $h$ or rebuild the alias structures.  In particular, each candidate point can have a pointer back to the alias structure within its partition cell, so it can replace its representative candidate point.  Moreover, since the points have weight proportional to their cell in the partition $W_j$, these weights do not change on a replacement. And more importantly, the points we never inspected to see if they belong to $h$ do not need to be replaced.  
This deterministic process only inspects at most $2$ points, and these can be replaced in $O(1)$ time.  Hence extending to $k$ samples, only increases the total runtime by an additive term $O(k)$.

\subparagraph*{Final bound.}

We now have the ingredients for our final bound.  In general the argument follows that of Lemma~\ref{lem:general} 
except for a few changes.  
First, we allow items with probability total less than $\gamma n$ to be ignored.  This replaces a $\log n$ factor in query time to be replaced with $t=\log (n/\gamma)$ term.  
Second, we require $O(t \log n)$ time to identify the relevant subset $X'$ in each of $t$ shallow-cutting structures. 
Then we select $O(1/\eps^3)$ candidate points from each of $t$ weight partitions, in $O(t/\eps^3)$ time.  Then pulling $k$ independent samples, and refilling the candidates takes $O(k)$ time.  
%We now that the weights by which the points are approximated are distorted by a $(1+\eps)$ factor $3$ times (by trimming the $1/\gamma n$ tail, but treating all points within $(1+\eps)$ as equi-weight, and the worst case sampling using a $1/\eps^3$-partition.  Adjusting each factor $\eps$ by a constant, e.g., to $\eps/10$ for $\eps < 1/2$ will ensure that the total distortion is still at most $(1+\eps)$.  
This results in the following final bound:

\begin{theorem}
Let $X$ be a set of $n$ weighted points in $\mathbb{R}^3$, where the smallest weight is $1$ and the largest weight is $U$.  Choose $0 < \gamma < \eps \leq 1/2$.  We can store $X$ in a data structure of size $O(n \min\{\log n, \log \log_n U\})$ such that for any integer $k$, we can extract $k$ independent random samples from the points $(\eps,\gamma)$-approximately proportional to their weights in worst case time 
$O(\log (n/\gamma) (\log n + 1/\eps^3) + k)$
\end{theorem}
In particular, if $\eps,\gamma = \Omega(1)$, then the worst case time is $O(\log^2 n + k)$ matching the expected time algorithm to sample by the exact weights.

% !TeXroot=range-sampling.tex

\section{Lower Bound for Worst-Case Time with Exact Weights}
\label{sec:lb}
In this section, we focus on proving a (conditional) lower bound for a data structure that can extract one random sample in $Q(n)$ worst-case
time using $S(n)$ space.
Our main result is that under a reasonably restricted model, the data structure must essentially solve an equivalent range searching problem
in the ``group model''. As a result, we get a conditional lower bound as this latter problem is conjectured to be difficult.
As an example, our conditional lower bound suggests that halfspace range sampling queries would require that 
$S(n) Q^3(n) = \Omega(n^3)$, i.e., with near-linear space we can only expect to get close to $O(n^{2/3})$ query time. 
This is in contrast to the case when we allow expected query time or approximate weights. 
As already shown, we can solve the same problem with $O(n \log n)$ space and $O(\log^2 n)$ query time which reveals a large polynomial gap
between the expected and the worst-case variants of the problem.
To our knowledge, this is the first time such a large gap appears in the range searching literature between the worst-case and expected query times. 

%\subsection{The Preliminaries}
\subparagraph{The Model of Computation.}
We assume the input is a list $X$  of $n$ elements, $x_1, \cdots, x_n$, where each element $x_i$ is associated with a real-valued weight $w(x_i)$. 
We use the decision tree model of computation. 
We assume the data structure has three components: a preprocessing algorithm that builds the data structure, a data structure which is
a set of $S(n)$  stored real values, and finally a query algorithm that given a query $q$ it returns an element sampled with the correct probability
in  $Q(n)$ worst-case time.
We allow no approximation: the query algorithm should return an element $x$ with exactly $w(x)/w(X)$ probability. 

\subparagraph{The main bottleneck.}
The challenge in obtaining our lower bound was understanding where the main computational bottleneck lied
and formulating a plan of attack to exploit it.
This turned out to be in the query algorithm. 
As a result, we place
no restrictions on the storage, or the preprocessing part of the algorithm, an
idea that initially sounds very counter intuitive. 
After giving the algorithm the input $X$ together with the weight assignment $w\colon X \to  \R$, the algorithm
stores some $S(n)$ values in its storage. 
Afterwards, we move to the query phase and this is where we would like to put reasonable limits.
We give the query algorithm the query $q$. 
We allow the query algorithm access to a set of $t$ real random numbers, $R= \{ r_1, \cdots, r_t\}$, generated uniformly in $[0,1]$. 
We restrict the query algorithm to be a ``binary decision tree'' $T$ but in a specialized format:
each  node $v$ of $T$ involves a comparison between some random number $r_v \in R$ and a rational function 
$f_v = g_v/G_v$ where $g_v$ and $G_v$ are $n$-variate polynomials of the input weights $w_1, \cdots, w_n$.
To be more precise, we assume the query algorithm can compute polynomials $g_v$ and $G_v$ 
(either using polynomials
stored at the ancestors of $v$ or using the values stored by the data structure).
The query algorithm at node $v$ computes the ratio $g_v$ and $G_v$
and compares it to the random value $r_v \in \{r_1, \cdots, r_t\}$.
If $v$ is an internal node, then $v$ will have two children $u_1$ and $u_2$ and the algorithm will proceed to $u_1$ if
$f_v(w(x_1), \cdots, w(x_n)) < r_v$ and to $u_2$ if otherwise. 
If $v$ is a leaf node, then $v$ will return (a fixed) element $x_v \in X$. 
Note that there is no restriction on the rational function $f_v$; it could be of an arbitrary size or complexity.
We call this the \mdef{separated Algebraic decision tree} (SAD) model since at
each node, we have ``separated'' the random numbers
from the rational function that involves the weights of the input elements; a more general
model would be to allow a rational function of the weights $w_i$ and the random numbers $r_j$.
If we insist the polynomials $g_v$ and $G_v$ be linear (i.e., degree one), then we call the
model \mdef{separated linear decision tree} model (SLD).

\begin{theorem}
  Consider an algorithm for range sampling in the SAD model where the input is a
  list of  $n$ of elements, $x_1, \cdots, x_n$ together with a weight assignment $w(x_i) \in \R$, for
  each $x_i$.
  Assume that the query algorithm has a worst-case bound, i.e., the maximum
  depth of its decision tree $T$ is bounded by a function $d(n)$.
  Then, for every query $q$, there exists a node $v \in T$ such that $G_v$ is divisible by the polynomial
  $\sum_{p \in q}w(p)$. 
\end{theorem}
\begin{proof}
  Consider the query decision tree $T$.
  By our assumptions, $T$ is a finite tree that involves some $N(n)$ nodes and it has the maximum  depth of $d(n)$.
  Each node $v$ involves a comparison between some random number $r_i$ and a rational function $f_v(w(x_1), \cdots, w(x_t))$.
  To reduce clutter, we will simply write the rational function as $f_v$.
  Let $\P$ be the set of all the rational functions stored at the nodes of $T$. 
  Note that each unique rational function appears only once in $\P$. 
  Consider the set $\Delta_\P = \{ f_1 - f_2 | f_1, f_2 \in \P\}$  which is the set of pairwise differences.
  As each unique rational function appears only once in $\P$, it follows that none of the rational functions in $\Delta_\P$ is identical to zero.
  This in particular implies that we can find real values $w_1, \cdots, w_n$ such that none of the rational functions in $\Delta_\P$ are zero
  on $w_1, \cdots, w_n$. 
  Thus, as these functions are continuous at the points $(w_1, \cdots, w_n)$, it follows that we can find a real value $\varepsilon > 0$ such that
  for every weight assignment $w(x_i) \in [w_i, w_i + \varepsilon]$, the rational functions in  $\Delta_\P$ have the same (none zero)
  sign. 
  In the rest of the proof, we only focus on the weight assignment functions $w$ with the property that
  $w(x_i) \in [w_i, w_i + \varepsilon]$.
  
  Let $\U$ be the unit cube in $\R^t$.
  $\U$ denotes the total probability space that corresponds to the random variables $r_1, \cdots, r_t$.
  Every point in $\U$ corresponds to a possible value for the random variables $r_1, \cdots, r_t$.
  Now consider one rational function $f_v$ stored at a node $v$ of $T$, and assume $v$ involves a comparison between $r_i$
  and $f_v$ and let $u_1$ and $u_2$ be the left and the right child of $v$. 
  By our assumption, we follow the path to $u_1$ if $f_v \le r_i$ but to $u_2$ if otherwise. 
  Observe that this is equivalent to partitioning $\U$ into two regions by a hyperplane perpendicular to the $i$-th axis
  at point $f_v$.
  As a result, for every node $v \in T$, we can assign a rectangular subset of $\U$ that denotes its \mdef{region} and it includes
  all the points $(r_1, \cdots, r_t)$ of $\U$ such that we would reach node $v$ if our random variables were sampled to be
  $(r_1, \cdots, r_t)$.

  The next observation is that we can assume the region of each node is defined by \mdef{fixed} rational functions.
  Consider the list of rational functions encountered on the path from $v$ to the root of $T$.
  Assume among this list, the rational functions 
  $f_1, \cdots, f_m$ are involved in comparisons with $r_i$.
  Clearly, the lower boundary of the region of $v$ along the $i$-th dimension is 
  $\min\{f_1, \cdots, f_m\}$ whereas its upper boundary is
  $\max\{f_1, \cdots, f_m\}$. 
  Now, observe that since we have assumed that each $f_i - f_j$ has a fixed sign, it follows that
  these evaluate to a fixed rational function.
  Thus, let $f_{i,v}$ (resp. $F_{i,v}$) be the rational function that describes the lower (resp. upper) boundary of the region of $v$ along the $i$-th dimension.
  Let $f_{i,v} = a_{i,v}/b_{i,v}$ and $F_{i,v} = A_{i,v}/B_{i,v}$ where $a_{i,v}, b_{i,v}, A_{i,v}, B_{i,v}$ are some polynomials of $w(x_1), \cdots, w(x_n)$
  stored in tree $T$ (e.g., $b_{i,v}$ is equal to some $G_u$ for some ancestor $u$ of $v$ and the same holds for $B_{i,v}$).

  The Lebesgue measure of the region of $v$, $\vol(v)$, is thus defined by the rational function
  \[
    \vol(v) =\prod_{i=1}^t (F_{i,v}-f_{i,v}) = \prod_{i=1}^t (\frac{A_{i,v}}{B_{i,v}} - \frac{a_{i,v}}{b_{i,v}}) = \prod_{i=1}^t \frac{A_{i,v} b_{i,v} - a_{i,v}B_{i,v}}{B_{i,v} b_{i,v}} .
  \]
  $\vol(v)$ is the probability of reaching $v$.
  Consider a query $q$ that contains $k$ elements.
  W.l.o.g, let $x_1, \cdots, x_k$ be these $k$ elements. 
  Let $v_1, \cdots, v_\ell$ be the set of all the leaf nodes that return $x_1$.
  For $x_1$ to have been sampled with correct probability we must have the following identity
  \[
    \sum_{i=1}^\ell \vol(v_i) = \frac{w(x_1)}{w(x_1) + \cdots + w(x_k)}.
  \]
  Observe that since the polynomial $w(x_1) + \cdots + w(x_k)$ is irreducible, it follows that 
  at least one of the polynomials $b_i$ or $B_i$ for some $i$ has this polynomial as a factor.
\end{proof}

\subparagraph{Conditional Lower Bound.}
Intuitively, our above theorem suggests that if we can perform range sampling in finite worst-case time,
then we should also be able to find the total weight of the points in the query range -- since the total weight $\sum_{p \in q} w(p)$ is encoded in some rational function.
This latter problem is conjectured to be difficult but obtaining provable good lower bounds  remains
elusive. 
This suggests that short of a breakthrough, we can only hope for a conditional lower bound.
This is reinforced by this observation that an efficient data structure for range sum queries leads to an efficient data structure for range sampling. 
\begin{observation}
  Assume for any set of $X$ of $n$ elements each associated with a real-valued weight,
  we can build a data structure that uses $S(n)$ space such that given a query range $q$,
  it can output the total weight of the elements in $q$ in $Q(n)$ time.

  Then, we can extract $k$ random samples from $q$ by a data structure that uses $O(S(n)\log n)$ space and  has the query time of
  $O(k Q(n) \log n)$. 
\end{observation}
\begin{proof}
  Partition $X$ into two equal-sized sets $X_1$ and $X_2$ and build the data structure for range sum on each.
  Then recurse on $X_1$ and $X_2$. 
  The total space complexity is $O(S(n) \log n)$.
  Given a query $q$, it suffices to show how to extract one sample.
  Using the range sum query $q$ on $X_1$ and $X_2$, we can know the exact value of $w_1 = \sum_{x \in q \cap X_1}w(x)$ and
  $w_2 = \sum_{x \in q \cap X_2}w(x)$.
  Thus, we can recurse into $X_1$ with probability $w_1/(w_1+w_2)$ and into $X_2$ with  $w_2/(w_1+w_2)$.
  In $O(\log n)$ recursion steps we find a random sample with query time
 $O(Q(n) \log n)$.
\end{proof}

This implies that in the SLD model, the range sampling problem in the worst-case is equivalent to the
range sum problem, ignoring polylog factors. 
This has consequences for query ranges like halfspaces where the latter problem is conjectured to be hard.
\begin{conj}\label{conj1}
    For every integer $n$,
    there exists a set $P$ of $\Theta(n)$ points in $\R^d$ with the following property:
    If for any function $w\colon P \to \R$ given as input, we can build a data structure of size $S(n)$
    such that for any query halfspace $h$, it can return the value $\sum_{p \in P \cap h}w(p)$ in
    $Q(n)$ time, then, we must have $S(n)Q^d(n) = \Omega(n^{d-o(1)})$.
\end{conj}

\begin{corollary}
  Assume Conjecture 1 holds for $d=3$.
  Then, there exists an input set of $n$ points in $\R^3$ such that 
  for any data structure that uses $S(n)$ space and solves the range sampling problem where the query algorithm is a 
  decision tree $T$ with worst-case query time $Q(n)$, we must have
  $S(n) Q^3(n) = \Omega(n^{3-o(1)})$.  
  Thus if space $S(n)$ is near-linear, query time $Q(n) = \Omega(n^{2/3 - o(1)})$.  
\end{corollary}

%%%%%%%%%%%%%%%%%%%%%%%%%%%%%%%%%%%%%%%%%%%%%%%%%%%%
%%%%%%%%%%%%%%%%%%%%%%%%%%%%%%%%%%%%%%%%%%%%%%%%%%%%
%%%%%%%%%%%%%%%%%%%%%%%%%%%%%%%%%%%%%%%%%%%%%%%%%%%%

%%%%%%%%%%%%%%%%%%%%%%%%%%%%%%%%%%%%%%%%%%%%%%%%%%%%

\bibliography{bib}
\bibliographystyle{plainurl}

%\appendix
%\input{app}

\end{document}